\newcommand{\R}{\mathbb{R}}

\let \Ra \Rightarrow

\let \Lra \Leftrightarrow

\documentclass[runningheads]{llncs}
\usepackage[utf8]{inputenc}
\usepackage{amsfonts, amsmath, amssymb} 
\PassOptionsToPackage{hyphens}{url}\usepackage{hyperref}
\usepackage{cleveref}
\usepackage{cite}

\DeclareMathOperator{\sign}{sign}

\begin{document}

\title{Polynomial Threshold Functions for Decision Lists}
\author{Vladimir Podolskii\inst{1,2} \and Nikolay V. Proskurin\inst{3}}
\institute{Courant Institute of Mathematical Sciences, New York University, New York, USA \and Steklov Mathematical Institute of Russian Academy of Sciences, Moscow, Russia \and HSE University, Moscow, Russia}
\maketitle

\begin{abstract}
For $S \subseteq \{0,1\}^n$ a Boolean function $f \colon S \to \{-1,1\}$ is a polynomial threshold function (PTF) of degree $d$ and weight $W$ if there is a polynomial $p$ with integer coefficients of degree $d$ and with sum of absolute coefficients $W$ such that $f(x) = \sign p(x)$ for all $x \in S$.
We study a representation of decision lists as PTFs over Boolean cubes $\{0,1\}^n$ and over Hamming balls $\{0,1\}^{n}_{\leq k}$.

As our first result, we show that for all $d = O\left( \left( \frac{n}{\log n}\right)^{1/3}\right)$ any decision list over $\{0,1\}^n$ can be represented by a PTF of degree $d$ and weight $2^{O(n/d^2)}$. This improves the result by Klivans and Servedio~\cite{DBLP:journals/jmlr/KlivansS06} by a $\log^2 d$ factor in the exponent of the weight. Our bound is tight for all $d = O\left( \left( \frac{n}{\log n}\right)^{1/3}\right)$ due to the matching lower bound by Beigel~\cite{beigel}.

For decision lists over a Hamming ball $\{0,1\}^n_{\leq k}$ we show that the upper bound on weight above can be drastically improved to $n^{O(\sqrt{k})}$ for $d = \Theta(\sqrt{k})$. We also show that similar improvement is not possible for smaller degrees by proving the lower bound $W = 2^{\Omega(n/d^2)}$ for all $d = O(\sqrt{k})$.
\end{abstract}

\section{Introduction}

The main object of studies in this paper are polynomial threshold functions.

\begin{definition}
    For $S \subseteq \{0,1\}^n$ a Boolean function $f \colon S \to \{-1,1\}$ is called a polynomial threshold function (PTF) of degree $d$ if there is an integer polynomial $p$ of degree $d$ such that $f(x) = \sign p(x)$ for all $x \in S$. If $p$ has integer coefficients, then the weight of a PTF is defined as a sum of absolute values of coefficients in $p$. 
\end{definition}

Polynomial threshold functions have been studied
intensively for decades. Much of this work was motivated by questions
in computer science~\cite{SIC:Saks93}, and PTFs are now an important
object of study in areas such as Boolean circuit
complexity~\cite{SICOMP:BS92,CC:GHR92,CCC:Beigel93,CC:KrausePudlak98},
learning theory~\cite{JCSS:KODS04,JCSS:KS04,DiakonikolasOSW11,BhattacharyyaGS18,DiakonikolasKS18a}, and communication complexity~\cite{EATCS:Sherstov08}.

In this paper, we study PTFs for the class of \emph{decision lists}. 
\begin{definition}
    A decision list $L$ on variables $x_{1},\ldots, x_n$ is a sequence of $h$ pairs and a bit 
    $$
    (\ell_1, b_1), (\ell_2,b_2),\ldots, (\ell_h,b_h), b_{h+1},
    $$
    where for all $i$ $\ell_i$ is a literal (either a Boolean variable $x_j$ or its negation) and $b_i \in \{-1,1\}$. On the input $x \in \{0,1\}^n$ the value of $L(x)$ is equal to $b_i$ if $i$ is the minimal index for which $\ell_i$ is true and if all $\ell_i$ are false, the value of $L(x)$ is equal to $b_{h+1}$.
\end{definition}
Decision lists have been widely studied in computational complexity~\cite{DBLP:books/cu/10/AnthonyCH10b, DBLP:journals/ipl/Blum92, DBLP:journals/cjtcs/ChattopadhyayMM20, DBLP:journals/jmlr/KlivansS06}. Partially, their importance stems from computational learning theory where attribute efficient learning of decision lists is an important research direction. One of the approaches to learning decision lists is through their representation as PTFs, to which several known learning algorithms can be applied~\cite{DBLP:journals/tcs/HellersteinS07, JCSS:KS04, DBLP:journals/jmlr/KlivansS06, DBLP:journals/combinatorica/ODonnellS10}. This constitutes one of the reasons to study PTF representations of decision lists. Decision lists turned out to be important in other areas of theoretical computer science as well. They provide a good source of lower bounds in the studies of threshold functions, threshold circuits, oracle computations and in other fields~\cite{beigel,DBLP:conf/coco/BuhrmanVW07,DBLP:conf/focs/ChattopadhyayM18, DBLP:conf/sofsem/UchizawaT15}.

A lot of effort has been put into the studies of effective PTF representations of Boolean functions, and for decision lists in particular. Klivans and Servedio \cite{DBLP:journals/jmlr/KlivansS06} have shown that for every $h < n$ any decision list can be computed by a PTF of degree $O(\sqrt{n}\log{n})$ and weight $ 2^{O(\frac{n}{h} + \sqrt{h}\log^2{n})} $; in terms of degree $d$ it results in $ 2^{O(\frac{n\log^2{d}}{d^2} + d\log{d})} $ weight bound. They used it to create the first online learning algorithm for decision lists that is subexponential in both running time and sample complexity; namely, their algorithm runs in time $ n^{O(n^{1/3}\log^{1/3}{n})} $ and mistake bound $ 2^{O(n^{1/3}\log^{4/3}{n})} $.
Servedio, Tan and Thaler proved in \cite[Theorem 6]{DBLP:journals/jmlr/ServedioTT12} that for any $ n^{1/4} \leq d \leq n $ any decision list $L$: $ \{-1, 1\}^n \to \{-1, 1\} $\footnote{In this result, input Boolean variables range over $\{-1,1\}$.} on $n$ variables has a degree-$d$ PTF of weight $ 2^{\tilde{O}_n((n/d)^{2/3})} $, where $ \tilde{O}_n $ suppresses the polylog factor of $n$. This bound is weaker for small $d$, but gives an upper bound for $d\geq \Omega(\sqrt{n})$. 
As for the lower bounds, Beigel~\cite{beigel} provided a decision list that requires weight of $ 2^{\Omega(n/d^2)} $ to be computed by a degree-$d$ PTF (we provide the details in Section~\ref{full}). 
Servedio, Tan and Thaler~\cite{DBLP:journals/jmlr/ServedioTT12} also proved a lower bound of $ 2^{\Omega(\sqrt{n/d})} $ for $ d = o\left( \frac{n}{\log^2{n}} \right) $, which is stronger than the bound of Beigel for $ d = \Omega(n^{1/3}) $.

Above we have discussed PTF representations of decision lists over the Boolean cube $\{0,1\}^n$. However, representations over its subsets are also relevant. The case of the Hamming ball $\{0,1\}^n_{\leq k}$ consisting of all vectors with at most $k$ 1s received some attention. Long and Servedio~\cite{DBLP:journals/siamdm/LongS14} gave bounds for the weights of PTFs for degree $d=1$. Their main motivation to study this setting comes from learning theory: in scenarios involving learning categorical data the common representation for examples is the one-hot encoded vector, which might have an extremely large amount of features, but only a small fraction of them can be active at the same time. The Winnow algorithm used in~\cite{DBLP:journals/jmlr/KlivansS06} can be applied to learn functions not only on the Boolean cube but also on its subsets, including $\{0,1\}^n_{\leq k}$, so it makes sense to study not only linear representations but also polynomial ones. Boolean functions on Hamming balls are also important outside of learning theory. For example, approximation of such functions by polynomials of low weight arises in the problem of indistinguishability~\cite{DBLP:conf/icalp/BogdanovW17, DBLP:journals/toct/HuangV22}, and recently the approximation on Hamming Balls in general received a lot of attention as a stepping stone for the approximation of some important Boolean functions~\cite{DBLP:journals/toc/BunKT20,Sherstov20}, including constant-depth circuits (AC$^0$)~\cite{DBLP:journals/siamcomp/BunT20,DBLP:journals/toc/BunT21,DBLP:conf/stoc/Sherstov22,DBLP:conf/stoc/SherstovW19}.

\paragraph*{Our results}

First, we address the PTF representations of decision lists over $\{0,1\}^n$. In our first result, we show that for every $h < n$ any decision list can be computed by a PTF of degree $O(\sqrt{n})$ and weight $ 2^{O(\frac{n}{h} + \sqrt{h}\log{n})} $. This improves the result of~\cite{DBLP:journals/jmlr/KlivansS06} by a logarithmic factor both in degree and exponent of weight. As a corollary, we slightly improve upon the upper bound for the attribute efficient learning of decision lists: we prove that there is an algorithm that learns decision lists in time $ n^{O(n^{1/3})} $ and mistake bound $ 2^{O(n^{1/3}\log{n})} $. 

In terms of degree vs. weight tradeoff, our results imply that for all $d = O\left( \left( \frac{n}{\log n}\right)^{1/3}\right)$ any decision list over $\{0,1\}^n$ can be computed by a PTF of degree $d$ and weight $2^{O(n/d^2)}$. The analogous result that follows the construction in~\cite{DBLP:journals/jmlr/KlivansS06} gave the weight upper bound of $2^{O((n \log^2 d)/d^2)}$. Our bound is tight for all $d = O\left( \left( \frac{n}{\log n}\right)^{1/3}\right)$ due to the matching upper bound by Beigel~\cite{beigel} mentioned above.

Clearly, our upper bound for the Boolean cube works also for decision lists over the Hamming ball $\{0,1\}^n_{\leq k}$ as they are just restrictions to a smaller domain. However, the tightness of this bound for the case of Hamming balls has to be investigated. We actually show that the upper bound can be improved to a polynomial in $n$ for $d$ roughly equal to $\sqrt{k}$. More precisely, we show that for $d = \Theta(\sqrt{k})$ any decision list over $\{0,1\}^n_{\leq k}$ can be represented with weight $n^{O(\sqrt{k})}$. However, we show that such an improvement is impossible for smaller degrees. In order to do so, we extend Beigel's lower bound to this setting, showing that there is a decision list that requires weight $2^{\Omega(n/d^2)}$ when computed by a PTF of degree $d = O(\sqrt{k})$.

\paragraph*{Our techniques}

To prove our upper bounds we adopt the same strategy as in~\cite{DBLP:journals/jmlr/KlivansS06}. We first represent a decision list as a sum of sublists, then we pointwise approximate each of the sublists. The improvement comes from the better approximation technique originated by Sherstov \cite{Sherstov20}, which we adapt for decision lists.

For the lower bound, we extend the proof of Beigel~\cite{beigel} to the setting of low-weight inputs. We use the same proof strategy of inductively constructing a sequence of inputs on which the value of the polynomial in PTF grows exponentially. The new ingredient in the proof is to keep the number of 1's in the input low by reusing them from the previous blocks.

\paragraph*{Organization}

In \cref{Preliminaries} we provide the necessary definitions and theorems. In \cref{full} we give our PTF construction for decision lists on $ \{0, 1\}^n $, and in \cref{limited} we prove the upper and the lower bounds for representing decision lists as a PTF on $ \{0, 1\}^n_{\leq k} $.

\section{Preliminaries} \label{Preliminaries}
We use the following notation for the \emph{arithmetization} of a literal $ \ell $:
\[
\tilde{\ell} = 
\begin{cases}
    x, \quad \ell \text{ is an unnegated variable } x, \\
    1 - x, \quad \ell \text{ is a negated variable } \overline{x}, \\
\end{cases}
\]
We also use the notations $ [n] = \{1, 2, \ldots, n\} $ and $|||p|||$ for the sum of absolute values of coefficients in $p$ (therefore, the weight of a PTF $p$ is equals to $|||p|||$ if $p$ has integer coefficients).

For technical reasons, we need a modified definition of a decision list, where the last $b_i$ is always equal to 0 instead of $ \pm 1 $. We call those decision lists \emph{modified}.

As mentioned in the introduction, one of the technical steps involves the pointwise approximation of Boolean functions. 
\begin{definition}
    Given a Boolean function $ f $: $ X \to \{0, 1\} $, we say that a polynomial $p$ is $ \varepsilon $-approximates the function $f$ iff $ \max_{x \in X} | f(x) - p(x) | \leq \varepsilon $. The $ \varepsilon $-approximate degree $ \deg_\varepsilon(f) $ is the minimum degree such that there exists a polynomial of degree $d$ that $\varepsilon$-approximates for the function $f$.
\end{definition}

Polynomial threshold functions may be viewed as a generalization of $ \varepsilon $-approximators: if $ p $ $\varepsilon$-approximates the function $f$ on $ X \subseteq \{0, 1\}^n $, then $ p(x) - \frac{1}{2} $ is a PTF that computes $f$ on $X$.

One of the commonly used families of polynomials for the Boolean function approximation is the Chebyshev polynomials of the first kind.
\begin{definition}
    The Chebyshev polynomials of the first kind are defined by the recurrence relation $ T_d(x) = 2xT_{d-1}(x) - T_{d-2}(x) $ with $ T_0(x) = 1 $ and $ T_1(x) = x $. 
\end{definition}

The solution of the recurrence above gives the following equation.
\begin{equation} \label{chebyshev}
    T_d(x) = \frac{1}{2} \bigg( \Big(x-\sqrt{x^2-1} \Big)^d + \Big(x+\sqrt{x^2-1} \Big)^d \bigg).
\end{equation}

It is obvious from the definition that $ \deg(T_d) = d $ and $ T_d $ has integer coefficients. We also need the following claims about the Chebyshev polynomials.
\begin{claim}[{\cite[Proposition 2.6]{Sherstov20}}] \label{chebyshev_upper}
    For every $ \delta > 0 $, $ T_d(1 + \delta) \geq 1 + d^2 \delta $.
\end{claim} 
\begin{claim} \label{chebyshev_lower}
    For every $ \delta \in [0; 2] $, $ T_d(1 + \delta) \leq (1 + \sqrt{\delta})^{2d} $.
\end{claim}
\begin{proof}
    Note that $ \sqrt{(1 + \delta)^2 - 1} = \sqrt{\delta(\delta + 2)} \leq 2\sqrt{\delta} $ for $ 0 \leq \delta \leq 2 $. By substituting $ x = 1 + \delta $ into \eqref{chebyshev}, we get
    \[ T_d(1 + \delta) = \frac{1}{2} \bigg( \Big( 1 + \delta - \sqrt{(1 + \delta)^2 - 1} \Big)^d + \Big( 1 + \delta + \sqrt{(1 + \delta)^2 - 1} \Big)^d \bigg) \leq \]
    \[ \leq \Big( 1 + \delta + \sqrt{(1 + \delta)^2 - 1} \Big)^d \leq (1 + \delta + 2\sqrt{\delta})^d = (1 + \sqrt{\delta})^{2d}. \]
\end{proof}
\begin{claim}[{{\cite[Section 2.2, Property 5]{DBLP:conf/icalp/BogdanovW17}}}] \label{chebyshev_weight}
    $|||T_d||| \leq 2^{2d}$.
\end{claim}

Klivans and Servedio used the Expanded-Winnow algorithm \cite[Theorem 2]{DBLP:journals/jmlr/KlivansS06} to derive the bounds on learning the decision lists class in the attribute-efficient learning model. It essentially runs the Winnow algorithm \cite[Algorithm 4]{DBLP:journals/ml/Littlestone87} on a set of all possible monomials of degree up to $d$. We note that the Winnow algorithm can learn Boolean functions on any domain $ X \subseteq \{0, 1\}^n $, and therefore we can apply the same bounds from the Expanded-Winnow to learn the decision lists on $ \{0, 1\}^n_{\leq k} $. Thus, the next theorem implicitly follows from the Expanded-Winnow algorithm.
\begin{theorem} \label{learning}
    Let $C$ be a class of Boolean functions over $S \subseteq \{0,1\}^n$ with the property that each $f \in C$ has a PTF of degree at most $d$ and weight at most $W$. Then there is an online learning algorithm for $C$ which runs in $n^d$ time per example and has mistake bound $ O(W \cdot d^2 \cdot \log{n})$.
\end{theorem}

\section{Decision lists on the Boolean cube} \label{full}
In this section, we prove the upper bound for representing decision lists as a PTF over the Boolean cube. Our proof can be viewed as an improvement of the proof of the following theorem by Klivans and Servedio. 
\begin{theorem}[{\cite[Theorem 7]{DBLP:journals/jmlr/KlivansS06}}]
    Let $ L $ be a decision list of length $n$ on $ \{0, 1\}^n $. Then for any $ h < n $, $L$ is computed by a PTF of degree $ O(\sqrt{h}\log{h}) $ and weight $ 2^{O(n/h + \sqrt{h}\log^2{h})} $.
\end{theorem}

We start with the outline of their construction and then proceed to prove the missing pieces to tighten up the upper bound. This construction consists of two parts. The ``outer'' part is the decomposition of the original decision list into the sum of modified sublists, i.e. we represent the $ L(x) $ as the $ L(x) = \sum_{i}^{n/h} a_i L_i(x) $, where each sublist is responsible for the independent block of size roughly $ h $. The ``inner'' part is the pointwise approximation of the sublist. To achieve a close enough approximation, Klivans and Servedio represent each of the sublists as a sum of conjunctions and approximate each of the conjunctions separately. The degree-weight traidoff is adjusted by the parameter $h$.

The problem with this approach is that in order to approximate the sublist, one has to very closely approximate the inner conjunctions, i.e. with a precision of at least $ O(\frac{1}{h}) $. But it was proved in~\cite{DBLP:journals/qic/Wolf10} that such an approximation requires the degree of $ \Theta(\sqrt{h\log{h}}) $, and while we can improve on the logarithmic factor, we can not get rid of it completely. However, we can notice that the set of conjunctions we are dealing with is not entirely random; in fact, every two successive conjunctions share most of their literals. That means we can adjust our approximator so it would increase it's precision as more terms in conjunction are set to zero. Luckily for us, it was proved in \cite{Sherstov20} that such an adjustment can be done with the same degree as in the regular approximator. In the next theorem, we adapt it for our needs and also prove the weight bound for it.

\begin{theorem} \label{univariate}
    For every $ d = \Theta(1)$ and every $\varepsilon = \Theta(1) $ there exists an univariate polynomial $ P_{n,d,\varepsilon}(x) $ such that $ \deg(P_{n,d,\varepsilon}) = O(\sqrt{n}) $, $ |||P_{n,d,\varepsilon}||| = 2^{O(\sqrt{n})} $ and
    \[ \forall t \in [0; 1] \quad |P_{n,d,\varepsilon}(t) - 1| \leq \varepsilon, \]
    \[ \forall t \in (2; n] \quad |P_{n,d,\varepsilon}(t)| \leq \frac{\varepsilon}{t^d}. \]
    
    Moreover, there exists a constant $ C = n^{O(\sqrt{n})} $ such that $ C \cdot P_{n,d,\varepsilon}(x) $ has integer coefficients.
\end{theorem}
\begin{proof}
    We can get such a polynomial from \cite[Theorem 3.7]{Sherstov20} (in fact, our statement is much more limited than the original one) with a proven degree bound but with no weight estimate. So it suffices for us to prove that the claimed weight bound holds as well.
    
    The polynomial $ P_{n,d,\varepsilon}(x) $ from \cite[Theorem 3.7]{Sherstov20} is given as $ P_{n,d,\varepsilon}(x) = p_1(x)^dp_2(p_1(x))p_3(x) $, where 
    \[ p_1(t) = \frac{T_{d_1 + 1}\left( 1 - 2\frac{t - 1}{n - 1} \right)}{T_{d_1 + 1}\left( \frac{n + 1}{n - 1} \right)} \]
    \[ p_2(t) = \sum_{i = 0}^D \binom{i + d - 1}{i} (1 - t)^i  \]
    \[ p_3(t) = B_{d_3}\left( \frac{1}{e^7} T_{\sqrt{n}}\left( 1 + \frac{2 - t}{n} \right) \right), \quad B_{d_3}(t) = \sum_{i = \lceil 2.5e^{-7}d_3 \rceil}^{d_3} \binom{d_3}{i} t^i(1 - t)^i \]
    with $d_1 = \lfloor \sqrt{2(n - 1)} \rfloor, \ D = O\left( d + \log{\frac{1}{\varepsilon}} \right)$ and $ d_3 = O\left(\log{\frac{1}{\varepsilon}}\right) $.
    
    To start with,
    \[ T_{d_1 + 1}\left( \frac{n + 1}{n - 1} \right) \leq \left( 1 + \sqrt{\frac{2}{n - 1}} \right)^{2\sqrt{2(n - 1)}} = \left( 1 + \frac{2}{\sqrt{2(n - 1)}} \right)^{2\sqrt{2(n - 1)}} \leq e^4, \]
    where the first inequality is \cref{chebyshev_upper} and the last one is $ (1 + \frac{a}{x})^x < e^a $ for $ a > 0 $ and $ x > 1 $. On the other hand, by \cref{chebyshev_lower}
    \[ T_{d_1 + 1}\left( \frac{n + 1}{n - 1} \right) \geq 1 + 2(n - 1) \cdot \frac{2}{n - 1} = 5, \]
    so the scaling factor in $ p_1(t) $ is bounded by constants, and by \cref{chebyshev_weight} $|||p_1||| = 2^{O(\sqrt{n})}$. By multiplying $ p_1 $ by $ (n - 1)^{d_1 + 1} $, we clear the denominators from $ 1 - 2\frac{t - 1}{n - 1} $, and thus it would have integer coefficients. The same holds for $ T_{\sqrt{n}}\left( 1 + \frac{2 - t}{n} \right) $ and $ n^{\sqrt{n}} $. As for the $T_{d_1 + 1}\left( \frac{n + 1}{n - 1} \right)$ in $p_1$, we know that it is a rational number $ \frac{a}{b} $ with $ b \leq (n - 1)^{d_1 + 1} $, which is bounded by constants. Therefore, $ a = O((n - 1)^{d_1 + 1}) $, and to clear all the denominators from $ P_{n,d,\varepsilon}(x) $ it suffices to multiply it by $ C = n^{O(\sqrt{n})} $.
    
    Next, both $ p_2(t) $ and $ B_{d_3}(t) $ have integer and independent of $n$ coefficients, as well as independent of $n$ degrees, so their degrees and weights are constant. Finally, the product of a constant number of terms of degree $ O(\sqrt{n}) $ and weight $ n^{O(\sqrt{n})} $ has the same degree and weight bounds, and the overall weight bound holds. The bound for the constant $C$ holds for the same reasons.
\end{proof}

With the new approximator, we are ready to proceed to the main proof of this section.
\begin{theorem} \label{approximate}
    Let $ L $ be a modified decision list of length $h$ on $ \{0, 1\}^n $. Then for every $ \varepsilon = \Theta(1) $ $L$ can be $ \varepsilon $-approximated by a polynomial of degree $ O(\sqrt{h}) $ and weight $ n^{O(\sqrt{h})} $. Moreover, $ p(0^n) = 0 $.
\end{theorem}
\begin{proof}
    Consider the decision list $ L = (\ell_1, b_1), \ldots (\ell_h, b_h), 0 $. It is straightforward to check that $ L $ can be expressed as $ L(x) = \sum_{i = 1}^h b_i \ell_i \bigwedge_{j = 1}^{h - 1} \overline{\ell_j} $, where the term corresponding to $ b_i $ is non-zero iff $ \ell_i $ is the first condition satisfied in $L$.
    
    For every $i$ put $ T_i(x) = \bigwedge_{j = 1}^{h - 1} \overline{\ell_j(x)} $ and $A_i(x) = 3(i - 1) - 3\tilde{\ell_1} - \ldots - 3\tilde{\ell}_{i - 1} $. Notice that $ A_i(x) = 0 \Lra T_i(x) = 1 $ and $ A_i(x) \geq 3 $ otherwise. Moreover, for every $j < i$ such that $ \ell_j = 1 $ the value of $ A_i(x) $ increases by 3. By corollary, for every $ i < j $ we have $ \ell_i = 1 \Ra A_i(x) + 3 \leq A_j(x) $ because $ A_j $ consists of every zero literal of $ T_i $ and $ \overline{\ell_i} = 0 $.
    
    Now consider $ p_i(x) = P_{3h, 2, \varepsilon/2}(A_i(x)) $ where $ P_{3h, 2, \varepsilon/2} $ is from \cref{univariate}. We will prove that $ p(x) = \sum_{i = 1}^h b_i \tilde{\ell_i} p_i(x) $ is the desired polynomial. To start with, if $ \ell_i = 0 $ then the corresponding term is also equals to zero; in particular, $ p(0^n) = 0 $. Next, let $ j_1 < \ldots < j_m $ be the set of all indexes such that $ \ell_{j_i} = 1 $. For $j_1$ we have $ A_{j_1}(x) = 0 $, and so $ \sign(b_{j_1} \tilde{\ell_{j_1}} p_{j_1}(x)) = b_{j_1} $ and $ |b_{j_1} \tilde{\ell_{j_1}} p_{j_1}(x)| \geq 1 - \varepsilon/2 $. Finally, for the remaining indexes we have $ 3 \leq A_{j_2}(x) < A_{j_3}(x) < \ldots < A_{j_m}(x) \leq 3h $, and so
    \[ \left|\sum_{i = 2}^m b_{j_i} \ell_{j_i} p_{j_i}(x)\right| \leq \sum_{i = 2}^m \left|p_{j_i}(x)\right| \leq \frac{\varepsilon}{2} \sum_{i = 3}^\infty \frac{1}{i^2} = \frac{\varepsilon}{2} \left( \sum_{i = 1}^\infty \frac{1}{i^2} - \frac{5}{4} \right) \leq \frac{\varepsilon}{2} \left( \frac{\pi^2}{6} - \frac{5}{4} \right) \leq \frac{\varepsilon}{2}. \]
    
    Therefore, for every $x \in \{0, 1\}^n$ we have $ |L(x) - p(x)| \leq \varepsilon $ and $ p(0^n) = 0 $. The bounds on degree and weight follows from the \cref{univariate} and the following observation: if $ p = a_d x^d + \ldots + a_1 x_1 + a_0 $ then
    \[ |||p(q(x))||| = |||a_d q(x)^d + \ldots + a_1 q(x) + a_0||| \leq |||a_d q(x)^d + \ldots + a_1 q(x)^d + a_0||| = \]
    \[ = ||| q(x)^d (a_d + \ldots + a_1 + a_0) ||| \leq |||q|||^d \cdot |||p|||\]
    
    By composing $ A_i $ of weight $ O(n) $ and degree $1$ with $ P_{3h, 2, \varepsilon/2} $ of weight $ 2^{O(\sqrt{h})} $ and degree $ \sqrt{h} $, we get a polynomial of degree $ \sqrt{h} $ and weight at most $ n^{O(\sqrt{h})} $.
\end{proof}

\begin{corollary} \label{ptf}
    Let $ L $ be a modified decision list of length $h$ on $ \{0, 1\}^n $. Then $L$ is computed by a PTF $p$ of degree $ O(\sqrt{h}) $ and weight $ h^{O(\sqrt{h})} $. Moreover, the $p(0^n) = 0$.
\end{corollary}
\begin{proof}
    Multiply the polynomial from the previous theorem by a constant $ C = n^{O(\sqrt{n})} $. By \cref{univariate}, every $p_i$ from the proof of \cref{approximate} now has integer coefficients. The resulting polynomial does not pointwise approximate the decision list anymore, but for every $x \in \{0, 1\}^n$ we have $ |CL(x) - p(x)| \leq C/\varepsilon $ and $ p(0^n) = 0 $. Therefore, for any sufficiently large constant $ \varepsilon $, say $ \frac{1}{100} $, $Cp$ is the desired PTF for $L$.
\end{proof}

Using \cref{ptf} as the inner approximator, we can use the outer construction by Klivans and Servedio to achieve the final PTF for the decision lists.
\begin{theorem} \label{optimal}
    Let $ L $ be a decision list of length $n$ on $ \{0, 1\}^n $. Then for any $ h < n $, $L$ is computed by a PTF of degree $ O(\sqrt{h}) $ and weight $ 2^{O(n/h + \sqrt{h}\log{h})} $.
\end{theorem}
The proof goes exactly the same as in \cite[Theorem 7]{DBLP:journals/jmlr/KlivansS06}, except the inner approximator has a constant precision, which does not affect the overall correctness. For the sake of completeness, we provide the proof in \cref{appendix}.
\begin{corollary} \label{optimal_deg}
    Let $ L $ be a decision list of length $n$ on $ \{0, 1\}^n $. Then for any $ d < \sqrt{n} $, $L$ is computed by a PTF of degree $ d $ and weight $ 2^{O(n/d^2 + d\log{d})} $.
\end{corollary}

In \cite{beigel}, Beigel proved the lower bound for the weight of PTF for the specific decision list called ODD-MAX-BIT.
\begin{definition}
    The ODD-MAX-BIT$_n$ function on input $ x \in \{0, 1\}^n $ is equal to $ (-1)^i $ where $i$ is the position of the rightmost 1 in $x$. If $ x = 0^n $ then ODD-MAX-BIT$_n(x) = 1$.
\end{definition}
\begin{theorem}[{{\cite{beigel}}}] \label{beigel}
    Let $p$ be a degree $d$ PTF with integer coefficients which computes ODD-MAX-BIT$_n$ on $ \{0, 1\}^n $. Then $|||p||| = 2^{\Omega(n/d^2)}$.
\end{theorem}

Note that \cref{optimal_deg} gives a PTF of weight $ 2^{O(n/d^2)} $ for any decision list and $d = O\left( \left( \frac{n}{\log n}\right)^{1/3}\right)$. Thus, in this range our bound is asymptotically optimal.

Using \cref{optimal}, we can also provide a slightly more efficient online learning algorithm for decision lists.
\begin{corollary}
    Let $ L $ be a decision list of length $n$ on $ \{0, 1\}^n $. Then $L$ is computed by a PTF of degree $ O(n^{1/3}) $ and weight $ 2^{O(n^{1/3}\log{n})} $.
\end{corollary}
\begin{proof}
    Apply \cref{optimal} with $ h = n^{1/3} $.
\end{proof}
\begin{corollary}
    There is an algorithm that learns decision lists on $ \{0, 1\}^n $ in time $ n^{O(n^{1/3})} $ and mistake bound $ 2^{O(n^{1/3}\log{n})} $. 
\end{corollary}
\begin{proof}
    Follow immediately from the previous corollary and \cref{learning}. 
\end{proof}

\section{Decision lists on Hamming balls} \label{limited}
In this section, we shift our focus from $ \{0, 1\}^n $ to $\{ 0, 1 \}^n_{\leq k}$. We show that both upper and lower bounds can be generalized on the new domain, and the degree-weight dependency is much more significant than it is on the Boolean cube. To be more precise, we show that if the degree parameter is low, then the Hamming ball scenario is not much different from the Boolean cube one, and we still need an exponential in terms of $n$ weight for a PTF. However, after a certain threshold it can be drastically improved. 

\subsection{Upper bound} \label{upper}
We first start with the upper bound; to achieve it, we can straightforwardly modify the proof of \cref{approximate} for the $ \{0, 1\}^n_{\leq k} $ domain.
\begin{theorem}
    Let $ L $ be a modified decision list of length $h$ on $ \{0, 1\}^n_{\leq k} $. Then for every $ \varepsilon = \Theta(1) $ $L$ can be $ \varepsilon $-approximated by a polynomial of degree $ O(\sqrt{k}) $ and weight $ n^{O(\sqrt{k})} $. Moreover, $ p(0^n) = 0 $.
\end{theorem}
\begin{proof}
    Recall the following definitions from the proof of \cref{approximate}
    \[ L(x) = \sum_{i = 1}^h b_i \ell_i \bigwedge_{j = 1}^{h - 1} \overline{\ell_j} = \sum_{i = 1}^h b_i \ell_i T_i(x) \qquad A_i(x) = 3(i - 1) - \tilde{\ell_1} - \ldots - \tilde{\ell}_{i - 1} \]
    
    We need to reexamine $ A_i(x) $. If for every $ x \in \{0, 1\}^n_{\leq k} $ we have $ A_i(x) > 0 $, then $ T_i(x) \equiv 0 $ on $ \{0, 1\}^n_{\leq k} $ and we can safely remove the corresponding term from the sum without affecting it's value on any input. Otherwise, let $ x' \in \{0, 1\}^n_{\leq k} $ be such an input that $ A_i(x') = 0 $. Flipping the value of one bit in $x'$ can affect at most one literal in $ T_i $, so if $ |x \oplus x'| = 1 $ then $ A_i(x) \leq 3 + A_i(x') $.
    
    Notice that because we are only interested in $ x \in \{0, 1\}^n_{\leq k} $, for any $ x $ we have $ |x \oplus x'| \leq |x \vee x'| \leq 2k $ and $ A_i(x) \leq 3 \cdot 2k + A_i(x') = 6k $. Thus, we can approximate $ A_i(x) $ with $ p_i(x) = P_{6k, 2, 2\varepsilon}(A_i(x)) $ instead of $ P_{3n, 2, \varepsilon/2} $ (where both polynomials are from \cref{univariate}), and this results in the desired bounds on degree and weight.
\end{proof}

\begin{corollary} \label{ptf_k}
    Let $ L $ be a decision list of length $n$ on $ \{0, 1\}^n_{\leq k} $. Then $L$ is computed by a PTF $p$ of degree $ O(\sqrt{k}) $ and weight $ n^{O(\sqrt{k})} $.
\end{corollary}
\begin{proof}
    The proof is absolutely the same as in \cref{optimal}.
\end{proof}

We note that the same ideas can be applied to \cite[Theorem 6]{DBLP:journals/jmlr/KlivansS06}, but it would result in a bound of $ \deg(p) = O(\sqrt{k}\log{n}) $ and $ |||p||| = 2^{O(\sqrt{k}\log^2{n})} $, which is much worse if $ k \ll n $.

As a side result, we also get an online-learning algorithm for decision lists on $ \{0, 1\}^n_{\leq k} $.

\begin{corollary}
    There is an algorithm that learns decision lists on $ \{0, 1\}^n_{\leq k} $ in time $ n^{O(\sqrt{k})} $ and with mistake bound $ n^{O(\sqrt{k})} $. 
\end{corollary}

\subsection{Lower bound} \label{lower}

While the previous proof is straightforward, it requires $d = \Theta(\sqrt{k})$. It turns out that it is not a coincidence: if we want to lower the degree even further, we would need to drastically increase the weight parameter. In fact, \cref{optimal_deg} is tight in the case of Hamming Balls as well as in the case of the Boolean Cube.

\begin{theorem} \label{bound}
    Let $p$ be a degree $d$ PTF with integer coefficients which computes ODD-MAX-BIT$_n$ on $ \{0, 1\}^n_{\leq k} $ and $d=O(\sqrt{k})$. Then $|||p||| = 2^{\Omega(n/d^2)}$.
\end{theorem}

The new proof is heavily based on the proof of \cref{beigel}; the main difference is the way we construct a polynomial that achieves a contradiction on its approximate degree.

\begin{proof}
    The proof goes as follows. We first partition $ [n] $ into blocks of even size about $ t = O(d^2) $. Note that we got $ r = \Omega(n/d^2) $ blocks in total. Then we prove that for every block $i$ we can find an input $y_i$ such that it may have 1's only in the first $i$ blocks and $ |p(y_i)| \geq 2^i $. If we succeed, then we get an input $ y_r $ such that $ |p(y_r)| \geq 2^r = 2^{\Omega(n/d^2)} $ and, as a corollary, $|||p||| = 2^{\Omega(n/d^2)}$. We will prove this claim by induction.
    
    We adjust the hidden constant in $d=O(\sqrt{k})$ so we can assume $k \geq t$. With that in mind, let $ y_0 = 1^t0^{n-t} $, i.e. the first $t$ bits are 1s and the remaining bits are 0s. Because $p$ has integer coefficients, we have $ |p(y_0)| \geq 1 $ and we are done with the base case. Now suppose without loss of generality that we have $ p(y_i) \geq 2^i $ (the case of negative $p(y_i)$ is completely analogous), we will prove that we can fill the $ (i + 1) $'s block in $ y_i $ and get $ y_{i + 1} $ such that $ p(y_{i + 1}) \leq -2^{i + 1} $.
    
    Let $ P(z) $: $ \{0, 1\}^{t/2} \to \R $ be constructed by the following constraints on the input of $p$.
    \begin{enumerate}
        \item Every block up to $i$'s is filled as in $ y_i $.
        \item In the $ (i+1) $'s block every odd bit becomes a new variable, and every other bit is set to 0.
        \item Let $ x_1, \ldots, x_t $ be the positive bits in the current input and let $z = (z_1, \ldots, z_{t/2})$ be the new set of variables. We change the value of every $x_i$ to $ 1 - z_i $.
        \item All the remaining bits are also set to 0.
    \end{enumerate}
    
    We first prove that for any $z$ the achieved input is indeed in $ \{0, 1\}^n_{\leq k} $. Note that because we started with $t$ 1s in $y_0$, every time we set $ z_i $ to 1 we change the corresponding $ x_i $ to 0. Thus, we never increase the number of 1s in the input, and because $ k \geq t $ it is in $ \{0, 1\}^n_{\leq k} $.
    
    Now suppose by contradiction that the desired $ y_{i+1} $ does not exist, i.e. for every $ x $ achieved by filling the $ (i+1) $'s block of $ y_i $, we have $ |p(x)| \leq 2^{i+1} $. By definition, $ P(0^{t/2}) = p(y_i) > 2^i $ but for every $ z \neq 0^{t/2} $ we have $-2^{i + 1} < p(z) < 0$ where the first inequality follows from $ |P(x)| \leq 2^{i + 1} $ and the second one follows from the definition of the ODD-MAX-BIT$_n$.

    We claim that the polynomial $ \tilde{P}(z) = -\frac{1}{3 \cdot 2^i}P(z) $ $\frac{1}{3}$-approximates the $ OR_{t/2} $ function. Indeed, $ \tilde{P}(0) = \frac{1}{3} $ and for every $ z \in \{0, 1\}^{t/2} \setminus \{0^{t/2}\} $ we have $ \frac{2}{3} \leq \tilde{P}(z) \leq \frac{4}{3} $. It is well known (see \cite{DBLP:journals/cc/NisanS94}) that $ \deg_{1/3}(OR_{n}) = \Omega(\sqrt{n}) $. For a suitable choice of the hidden constant in $t$ we get that $ \deg(\tilde{P}) > \sqrt{t} = d $, but at the same time $ \deg(\tilde{P}) = \deg(P) \leq \deg(p) = d $, which contradicts the non-existence of the desired $ y_{i+1} $.
\end{proof}

\section{Discussion}

Despite having optimal bounds on decision lists for almost all $ d = O(n^{1/3}) $, the situation is less clear for other values of $d$. As we mentioned in the introduction, Servedio, Tan and Thaler~\cite[Theorem 6]{DBLP:journals/jmlr/ServedioTT12} proved that for any $ n^{1/4} \leq d \leq n $ any decision list $L$: $ \{-1, 1\}^n \to \{-1, 1\} $ on $n$ variables has a degree-$d$ PTF of weight $ 2^{\tilde{O}_n((n/d)^{2/3})} $. This gives an upper bound for $ d = \Omega(n^{1/2}) $, but it uses $ \{-1, 1\}^n $ domain rather than $ \{0, 1\}^n $. They also proved a lower bound of $ 2^{\Omega(\sqrt{n/d})} $ for both domains and $ d = o\left( \frac{n}{\log^2{n}} \right) $, which is stronger than \cref{beigel} for $ d = \Omega(n^{1/3}) $. It uses the same approach, but a different technique for bounding the degree on a block, which suggests that a matching upper bound (if there is one) should have different construction than the one achieved by \cref{ptf}. We also note that because of the block approach this lower bound can be also adapted for Hamming Balls.

As for the approximation on $ \{ 0, 1 \}^n_{\leq k} $, an obvious open question is to generalize \cref{bound} for $ k = o(d^2) $. Another open question is to find optimal degree-weight tradeoffs for pointwise approximation of Boolean functions on both $ \{0, 1\}^n $ and $ \{0, 1\}^n_{\leq k} $, as suggested in \cite{DBLP:journals/toct/HuangV22}. We suspect that by including $ \varepsilon $ into analysis of \cref{univariate} it can be generalized for any $ \varepsilon = o(1) $ as well to achieve a similar to \cref{approximate} result. After that, the similar ideas used in \cite{DBLP:conf/icalp/BogdanovW17} and \cite{DBLP:journals/toct/HuangV22} can be applied to get degree-weight traidoff results for a class of decision lists.

\bibliographystyle{splncs04}
\bibliography{sample}

\appendix

\section{Ommited proofs}\label{appendix}

\begin{proof}[of \cref{optimal}]
    We can represent $L$ as follows:
    \[ L(x) = \sign C \left( \sum_{i=1}^{n/h} 3^{n/h - i + 1} f_i(x) + b_{n+1} \right). \]
    with $C$ from \cref{univariate}, where each $ f_i $ is a sublist responsible for the bits on the interval $ x_{(i - 1)h + 1}, \ldots, x_{ih} $, which returns zero if no literals on this interval is satisfied. It is straightforward to check the correctness of the above decomposition: if $f_i$ contains the first satisfied literal, then it determines the sign of the sum because $ 3^{n/h - i + 1} > \sum_{j < i} 3^{n/h - j + 1} $ and $ f_j(x) = 0 $ for every $ j < i$.
    
    Every sublist $ f_i $ is a modified decision list, so we can apply \cref{ptf} and replace every $ Cf_i(x) $ with $ p_i(x) $, resulting in
    \[ H(x) = \sum_{i=1}^{n/h} 3^{n/h - i + 1} p_i(x) + Cb_{n+1}. \]
    
    Our goal is to show that $ H(x) $ is the desired PTF. First of all, if $ x = 0^n $ then every $ p_i(x) = 0 $ and $ \sign{H(x)} = b_{n+1} $. Otherwise, let $ r = (i - 1)h + c $ be the first bit such that $ \ell_r $ is satisfied. We have several cases: 
    \begin{enumerate}
        \item If $ j < i $ then $ 3^{n/h - i + 1} p_j(x) = 0 $;
        \item $ 3^{n/h - i + 1} p_i(x) $ differs from $ 3^{n/h - i + 1} Cb_r $ by at most $ C3^{n/h - i + 1} \frac{1}{100} $;
        \item The magnitude of each of the remaining values is at most $ C3^{n/h - j + 1} \left(1 + \frac{1}{100}\right)$.
    \end{enumerate}
    
    Combining these bounds, the value of $ H(x) $ differs from $ 3^{n/h - i + 1} Cb_r $ by at most
    \[ C \left( 1 +  \frac{3^{n/h - i + 1}}{100} + \left( 1 + \frac{1}{100} \right) \sum_{j = 1}^{n/h - i} 3^j \right). \]
    
    The value of the sum is less than $ \frac{3^{n/h - i + 1}}{2} $, so the overall value is less than $ C3^{n/h - i + 1} $ and $ \sign{H(x)} = b_r $. The bounds on degree and weight follows from \cref{ptf}.
\end{proof}

\end{document}